\documentclass[11pt]{article}
\usepackage[T1]{fontenc}      
\usepackage[english]{babel}   

\usepackage[letterpaper,twoside,vscale=0.8,hscale=0.75,nomarginpar,hmarginratio=1:1]{geometry} 
\usepackage{setspace}         
\usepackage[bottom]{footmisc} 
\usepackage{framed}           
\usepackage{parskip}           

\usepackage{amsmath, amsthm, amssymb} 
\usepackage{mathtools}        
\usepackage{bm}               
\usepackage{bbm}              
\usepackage{nicefrac}         

\usepackage{graphicx}         
\usepackage{subcaption}       
\usepackage{wrapfig}          
\usepackage{array}            
\usepackage{float}            
\usepackage[font={small,it}]{caption} 
\usepackage{multirow}         

\usepackage{algorithm}        
\usepackage{algpseudocode}    

\usepackage[numbers,sort]{natbib} 

\usepackage{verbatim}         
\usepackage{xspace}           
\usepackage{enumitem}         
\usepackage{soul}             
\PassOptionsToPackage{dvipsnames}{xcolor}  
\usepackage{tikz-qtree, tikz-qtree-compat} 

\usepackage[backref,pageanchor=true,plainpages=false,pdfpagelabels,
bookmarks,bookmarksnumbered]{hyperref} 
\usepackage{cleveref}         

\usepackage{thmtools}
\usepackage{thm-restate}

\declaretheorem[name=Theorem,numberwithin=section]{theorem}
\declaretheorem[sibling=theorem]{lemma}

\declaretheorem[name=Corollary,sibling=theorem]{corollary}

\theoremstyle{definition}
\declaretheorem[name=Definition,sibling=theorem]{defn}
\declaretheorem[name=Definition,sibling=theorem]{definition}
\declaretheorem[sibling=theorem]{remark}

\declaretheorem[name=Observation,sibling=theorem]{observation}

\declaretheorem[name=Claim,sibling=theorem]{claim}

\ifx\proof\undefined
\newenvironment{proof}[1][\protect\proofname]{\par
	\normalfont\topsep6\p@\@plus6\p@\relax
	\trivlist
	\itemindent\parindent        
	\item[\hskip\labelsep\scshape #1]\ignorespaces
}{%
	\endtrivlist\@endpefalse
}
\providecommand{\proofname}{Proof}
\fi

\Crefname{algocf}{Algorithm}{Algorithms}
\crefname{algocfline}{line}{lines}
\Crefname{invariant}{Invariant}{Invariants}
\Crefname{claim}{Claim}{Claims}
\Crefname{subclaim}{Subclaim}{Subclaims}

\definecolor{fluorescentyellow}{rgb}{0.8,1.0,0.0}    
\sethlcolor{fluorescentyellow}                      

\newcommand{\eat}[1]{}  

\DeclareMathOperator{\per}{per}

\DeclareMathOperator{\perm}{perm}

\newcommand{\sumL}{\sum\limits}

\newcommand{\norm}[1]{\left\Vert#1\right\Vert}           

\newcommand{\del}{\backslash}
\newcommand{\RN}[1]{\textup{\uppercase\expandafter{\romannumeral#1}}} 

\makeatletter
\newcommand{\seqref}[1]{\textup{\tagform@split{\getrefnumber{#1}}}}
\newcommand\tagform@split[1]{%
	\begingroup
	\m@th\normalfont(\ignorespaces #1\unskip\@@italiccorr)%
	\endgroup
}
\makeatother

      \newcommand{\RR}{\mathbb{R}}

    \newcommand{\cF}{\mathcal{F}}

\usepackage{natbib}
\usepackage{fullpage}

\title{An Algorithmic Upper Bound for Permanents via a Permanental Schur Inequality}

\date{}
\author{Aditi Laddha \thanks{aditi.laddha@yale.edu; supported by the Institute for Foundations of Data Science at Yale University.}
\and Madhusudhan Reddy Pittu \thanks{madhusudhan.p@nyu.edu; supported in part by NSF grants CCF-2224718 and CCF-2422926.}}

\begin{document}

\maketitle
\begin{abstract}
Computing the permanent of a non-negative matrix is a computationally challenging, \#P-complete problem with wide-ranging applications. We introduce a novel permanental analogue of Schur's determinant formula, leveraging a newly defined \emph{permanental inverse}. Building on this, we introduce an iterative, deterministic procedure called the \emph{permanent process}, analogous to Gaussian elimination, which yields constructive and algorithmically computable upper bounds on the permanent. Our framework provides particularly strong guarantees for matrices exhibiting approximate diagonal dominance-like properties, thereby offering new theoretical and computational tools for analyzing and bounding permanents.
\end{abstract}

\section{Introduction}
The permanent of a matrix, despite its deceptively simple definition, is notoriously difficult to compute. It is well-known that the exact computation of the permanent is \#P-complete~\cite{valiant-permanet-hardness}, situating it at the forefront of complexity theory and establishing its computational intractability for all but trivially sized matrices. Despite this difficulty, permanents play a critical role across diverse areas such as combinatorics, graph theory (particularly in counting perfect matchings in bipartite graphs \cite{Minc_book-84}), quantum computing (specifically within boson sampling experiments aimed at demonstrating quantum advantage~\cite{AA11-boson-sampling}), and statistical physics (e.g., in dimer covering models \cite{HLLB08-monomer-dimer}).

Due to the permanent's computational complexity, significant effort has been directed towards deriving efficient upper bounds and approximation algorithms. Classical results, such as the Bregman-Minc inequality and its numerous refinements, form a rich and diverse body of work (e.g., \cite{Minc_binary63, Marcus-Minc_perm65, Bregman73, Schrijver_Minc-78, Minc_book-84, Schrijver_matching-98, HKM-98, LB-04, Samorodnitsky-08, Soules_minc-00, Soules_nonnegative-03, GS-14}). 

Recently, there has been growing interest in approximating the permanents of positive semidefinite (PSD) matrices \cite{anari2017simply, yuan2022maximizing, meiburg2023inapproximability, ebrahimnejad2025approximability}. Permanents of PSD matrices appear naturally in quantum optics and boson sampling in quantum computing \cite{aaronson2011computational,scheel2004permanents, shirai2003random}.

However, existing bounds have largely emerged from combinatorial, scaling, or probabilistic frameworks. In contrast, this work develops a deterministic, iterative procedure rooted in linear algebraic principles. Our main result is an algorithm for establishing upper bounds on the permanent of non-negative matrices and positive semidefinite matrices. The approach provides a new algorithmic pathway that is both theoretically sound and computationally feasible. We systematically adapt powerful and intuitive tools from determinant theory. 

The main contributions of this paper are summarized as follows.
\paragraph{A Permanental Inverse and Schur's Formula for Permanents.}
We define a novel analogue of the matrix inverse, constructed from permanents of submatrices. While it lacks multiplicative inverse properties, this \emph{permanental inverse} satisfies structural inequalities that enable us to derive new permanental inequalities and to establish our algorithmic upper bound on the permanent.
	\begin{defn}[Permanental Inverse]
		The \emph{permanental inverse} of a non-negative matrix $B\in \RR^{d\times d}_{\geq 0}$ with $\per(B)\neq 0$ is the matrix $C$ with entries
		\begin{align*}
			c_{i,j}=\frac{\per(B_{j,i})}{\per(B)},
		\end{align*}
		where $B_{j,i}$ is the matrix obtained by removing the $j^{\text{th}}$ row and $i^{\text{th}}$ column of $B$.
		We use $B^*$ to denote the permanental inverse of $B$.
	\end{defn}
	
Leveraging the permanental inverse, we establish an analogue of Schur's formula for determinants for permanents. This result provides the core theoretical engine for our upper bounds.
	\begin{restatable}[Permanental Schur's Formula]{theorem}{PermanentInequality}
		\label{thm:permanent-inequality}
		Let $A \in \RR^{n \times n}_{\geq 0}$ be a block matrix of the form
		\begin{align*}
		A = \begin{bmatrix}
			B & Y \\
			X^{\top} & W
		\end{bmatrix},
		\end{align*}
		where $B \in \RR^{d \times d}$ has non-zero permanent. Then the permanent of $A$ satisfies
		\begin{align}
			\label{eqn:perm-inv-prop} 
		\per(A) \leq \per(B) \cdot \per\left(W + X^{\top} B^* Y \right).
		\end{align}
	\end{restatable}
    See Section \ref{sec:perm-inv} and Section \ref{sec:schur-permanent}  for more details regarding the permanental inverse and Theorem \ref{thm:permanent-inequality}. 

\paragraph{A Constructive Algorithmic Upper Bound.} We introduce an iterative procedure, called the \emph{Permanent Process}, inspired by Gaussian elimination. This algorithm yields a provable upper bound on the permanent of any non-negative or PSD matrix, while being computationally efficient, requiring only $O(n^3)$ operations for an $n\times n$ matrix (Algorithm \ref{alg:perm-process}).
\begin{algorithm}[H]
\caption{The Permanent Process}\label{alg:perm-process}
\begin{algorithmic}
\State \textbf{Input.} $A \in \mathbb{R}^{n \times n}$, $A$ non-negative or PSD

\State $A^{(1)} \gets A$

\For{$t = 1 \text{ to } n-1$}
    \For{$i = t+1 \text{ to } n$}
        \For{$j = t+1 \text{ to } n $}
            \State $a_{i,j}^{(t+1)} \gets a_{i,j}^{(t)} + \dfrac{a_{i,t}^{(t)} \cdot a_{t,j}^{(t)}}{a_{t,t}^{(t)}}$
        \EndFor
    \EndFor
\EndFor
\State \textbf{Return} $A^{(n)}$
\end{algorithmic}
\end{algorithm}
Observe that the $i,j$-th entry of $A$ remains unchanged after step $\min(i,j)-1$ of the outer loop. In other words, for $t \geq \min(i,j)$, $A^{(t)}_{i,j} = A^{(t-1)}_{i,j}$. 

Our main technical result shows that the entries of the matrix returned by Algorithm \ref{alg:perm-process} can be used to upper bound the permanent of $A$ as follows.

    \begin{restatable}[Permanent Process]{theorem}{permprocess} \label{thm:perm-process-invar}
    Let $A$ be a real nonnegative matrix or a real PSD matrix.
    The permanent of $A$ is upper bounded by the product of the diagonal entries of $A^{(n)}$:
	\begin{align*}
		\per(A) \leq \prod_{1\leq i\leq n}a_{i,i}^{(i)} = \prod_{1\leq i\leq n}a_{i,i}^{(n)}.
	\end{align*}
    \end{restatable}

See Section \ref{sec:perm-process} for further details regarding the permanent process and its properties.
\begin{remark}
If, in the update rule $a_{i,j}^{(t+1)} \leftarrow a_{i,j}^{(t)} + \dfrac{a_{i,t}^{(t)} \cdot a_{t,j}^{(t)}}{a_{t,t}^{(t)}}$, we replace the ``$+$'' with a ``$-$'', then the determinant of $A$ is exactly equal to the product of the diagonal entries of $A^{(n)}$ (see Corollary \ref{cor:Gauss_elim-invar}).
\end{remark}
\begin{remark}
It is not immediately clear whether the representation size of the entries of $A^{(t)}$ remains polynomially bounded throughout the process. We establish that it does for non-negative matrix in Theorem \ref{thm:boundedness} in Section \ref{sec:boundedness}.
\end{remark}

\paragraph{Provable Guarantees for Structured Matrices:} We show that for matrices exhibiting a notion of approximate diagonal dominance, a structure common in numerical linear algebra and network models, our upper bound yields strong theoretical guarantees. See Section \ref{sec:theoretical-ub} for more details.
\begin{restatable}{theorem}{applicationbound}\label{thm:application}
Let \( A \in \RR^{n \times n}_{\geq 0} \) be a non-negative matrix satisfying
\begin{align*}
	\frac{(1+\varepsilon)^2}{\varepsilon} \sum_{s=1}^{\min(i,j)} \frac{a_{i,s} a_{s,j}}{a_{s,s}} \le a_{i,j}
\end{align*}
for some \( \varepsilon > 0 \). Then
\begin{equation*}
    \per(A) \le (1+\varepsilon)^n \cdot \prod_{i=1}^n a_{i,i}.
\end{equation*}
\end{restatable}

Taken together, these results offer a new algorithmic perspective on permanents and expand the analytical toolbox for bounding them, with potential applications in combinatorics, statistical physics, and quantum computation.

\paragraph{Organization of the paper.} We provide preliminaries and results that we will use in Section \ref{sec:preliminaries}. In Section \ref{sec:perm-schur}, we prove Theorem \ref{thm:permanent-inequality}. In Section \ref{sec:perm-process}, we prove Theorem \ref{thm:perm_process-invar} and discuss some applications of our framework. In Section \ref{sec:boundedness}, we show that the representation size of the matrix entries remains polynomially bounded with respect to the representation size of the input during Algorithm \ref{alg:perm-process}.

\section{Preliminaries} \label{sec:preliminaries}
\subsection{Notation}
To discuss matrices, we will use the following standard notation. Let $A$ be an $n \times n$ matrix with real-valued entries, denoted $A \in \mathbb{R}^{n \times n}$.

\begin{itemize}
	\item \textbf{Matrix Entries:} $(A)_{i,j}$ or $a_{i,j}$ (the corresponding lowercase letter) refers to the entry in the $i^{th}$ row and $j^{th}$ column of $A$.
	
	\item \textbf{Submatrices by Selection:} For index sets $S, T \subseteq \{1, \dots, n\}$, we denote by $A(S, T)$ the submatrix formed by taking the rows indexed by $S$ and columns indexed by $T$. By convention, $\det(A(\emptyset, \emptyset)) = \per(A(\emptyset, \emptyset)) = 1$. We will sometimes use $\det_A(S, T) := \det(A(S, T))$. 
	
\item \textbf{Submatrices by Deletion:} We use several notations for submatrices formed by deleting rows or columns.
\begin{itemize}
	\item The matrix $A_{-i,.}$ denotes the matrix obtained by deleting the $i^{th}$ row, and $A_{.,-j}$ denotes the matrix obtained by deleting the $j^{th}$ column.
	\item $A_{i,j}$ denotes the matrix obtained by deleting both the $i^{th}$ row and the $j^{th}$ column.
	\item For deleting multiple rows and columns, the notation $A(-S, -T)$ is shorthand for the submatrix formed by deleting the rows in set $S$ and columns in set $T$.
\end{itemize}

	\item \textbf{Entrywise Inequality:} The expression $A \ge B$ means that every entry in $A$ is greater than or equal to the corresponding entry in $B$ (i.e., $a_{i,j} \ge b_{i,j}$ for all $i,j$).
	
	\item \textbf{Functions:} For any function $f:[k]\rightarrow [d]$, let $\text{img}_{f,S}:= \{f(j): j\in S\}$ be the image of $S$ according to $f$; we simply write $\text{img}_{f}$ when $S=[k]$.
\end{itemize}

\subsection{Gaussian Elimination}

\begin{remark}[A Note on Convention]
	The method described here is a specific variant of Gaussian elimination designed to produce a \textbf{lower triangular} matrix. This is a non-standard convention, as the standard algorithm is typically defined to produce an \textit{upper triangular} matrix.
\end{remark}

Let $A^{(t)}$ denote the state of the matrix at the beginning of step $t$, with the initial matrix being $A^{(1)} = A$. The goal of Gaussian Elimination is to iteratively transform $A$ into a lower triangular matrix. The state of the matrix entries after the end of step $t$ for some $1\leq t\leq n-1$ is given by 
\begin{align}
	a_{i,j}^{(t+1)}=\begin{cases} a_{i,j}^{(t)}-\frac{a_{i,t}^{(t)}a_{t,j}^{(t)}}{a_{t,t}^{(t)}}  , &\text{for $j\geq t+1$}\\
		a_{i,j}^{(t)}, &\text{otherwise.}
	\end{cases}
\end{align}
In simpler terms, at each step $t$, this process uses the pivot element $a_{t,t}$ to create zeros in all entries to its right, within the same row $t$.
\begin{theorem}[Gaussian Elimination Invariant]
	\label{thm:Gaussian-determinant}
	The entries of the matrix $A^{(t)}$ are ratios of determinants of certain sub-matrices of $A$: 
	\begin{align*}
		a_{i,j}^{(t)} =\frac{\det_A\left([r-1]+\{i\}, [r-1]+\{j\}\right)}{\det_A([r-1],[r-1])}, \quad r=\min(j,t).
	\end{align*}
\end{theorem}

\begin{corollary}[Determinant Property]
	\label{cor:Gauss_elim-invar}
	Let $A^{(n)}$ be the final lower triangular matrix obtained after running the full elimination process on $A$. Then the product of its diagonal entries equals the determinant of $A$:
	\begin{equation*}
		\det(A) = \prod_{i=1}^{n} a_{i,i}^{(n)}.
	\end{equation*}
\end{corollary}

\subsection{Schur's Formula}

The Schur complement is a fundamental tool for working with the determinants of block matrices.

\begin{theorem}[Schur's Determinant Formula] 
	\label{thm:schur_formula}
	Let $A \in \mathbb{R}^{n \times n}$ be a block matrix of the form
	\begin{equation*}
		A = \begin{bmatrix}
			B & Y \\
			X^{\top} & W
		\end{bmatrix},
	\end{equation*}
	where $B \in \mathbb{R}^{d \times d}$ is an invertible matrix. Then the determinant of $A$ is given by
	\begin{equation*}
		\det(A) = \det(B) \cdot \det(W - X^{\top} B^{-1} Y).
	\end{equation*}
	The matrix $S = W - X^{\top} B^{-1} Y$ is called the \textbf{Schur complement} of $B$ in $A$.
\end{theorem}
\subsection{Permanent of PSD Matrices} \label{lem:psd-perm}
The following lemma from \cite{marcus1962inequalities} shows that the permanent of any PSD matrix is non-negative.
\begin{theorem}[\cite{marcus1962inequalities}]
    Let $A = V^\top V$ be an $n \times n$ PSD matrix and let $v_1, \ldots, v_n$ be the columns of $V$. Then the permanent of $A$ is given by
    \begin{equation*}
        \per(A) = \frac{1}{n!} \norm{\sum_{\sigma \in \mathcal{S}_n} v_{\sigma(1)} \otimes v_{\sigma(2)} \ldots \otimes v_{\sigma(n)}}^2.
    \end{equation*}
\end{theorem}
\noindent
Here, $\otimes$ denotes the tensor product. If each $v_i \in \mathbb{R}^d$, then 
$v_{1} \otimes \cdots \otimes v_{n}$ is an element of the $n$-fold tensor product space 
$(\mathbb{R}^d)^{\otimes n} \cong \mathbb{R}^{d^n}$, i.e., a vector in a $d^n$-dimensional space.

\section{Generalizing Determinantal Concepts for the Permanent} \label{sec:perm-schur}

This section extends classical determinantal concepts, such as the matrix inverse and Schur's formula, to the setting of the matrix permanent.

\subsection{The Permanental Inverse}
\label{sec:perm-inv}
\subsubsection{Motivation from the Determinant}
To start, recall that one way to define the inverse of an invertible matrix $B$ is using Cramer's rule, where the $i,j$-th entry of the inverse $C = B^{-1}$ is given by:
\begin{equation*}
	c_{i,j} = \frac{\det(B_{j,i})}{\det(B)}.
\end{equation*}
Here, $B_{j,i}$ is the submatrix of $B$ formed by removing row $j$ and column $i$. This definition naturally gives us $B^{-1}B = BB^{-1} = I$. This provides a direct template for defining a similar concept for the permanent.

\subsubsection{Definition}
\begin{definition}[Permanental Inverse]
	For a non-negative matrix $B \in \mathbb{R}^{d \times d}_{\ge 0}$ with $\per(B) > 0$, the \textbf{permanental inverse}, denoted $B^*$, is an $n\times n$ matrix with entries:
	\begin{equation*}
		(B^*)_{i,j} = \frac{\per(B_{j,i})}{\per(B)}.
	\end{equation*}
\end{definition}

\subsubsection{Crucial Differences and Properties}
Unlike the determinantal inverse, multiplication by $B^*$ does not typically recover the identity matrix. Instead, it satisfies a matrix inequality.

\begin{claim}
	\label{claim:perm-inverse-basic}
	For a non-negative matrix $B$, we have $B^*B \ge I$ and $BB^* \ge I$. In general, $B^*B \neq BB^*$.  
\end{claim}
\begin{proof}
Recall that $A \ge B$ means that every entry in $A$ is greater than or equal to the corresponding entry in $B$ (i.e., $a_{i,j} \ge b_{i,j}$ for all $i,j$).
One can evaluate the diagonal entries of $B^*B$ as
	\begin{equation*}
		(B^*B)_{ii} = \frac{1}{\per(B)} \sum_{j=1}^{d} b_{ji} \cdot \per(B_{j,i})
	\end{equation*}
	The sum $\sum_{j=1}^{d} b_{ji} \cdot \per(B_{j,i})$ is the Laplace expansion of the permanent of $B$ along column $i$, which equals $\per(B)$. Thus, the diagonal entries are $(B^*B)_{ii} = \frac{\per(B)}{\per(B)} = 1$.
	
	For the off-diagonal entries (where $k \neq i$), since $B$ is a non-negative matrix, all its permanents and entries are non-negative. Thus, every term in the sum for $(B^*B)_{ik}$ is non-negative, meaning $(B^*B)_{ik} \ge 0$.
	
	Combining these two points, the diagonal entries of $B^*B$ are 1 and the off-diagonal entries are non-negative. By definition, this means $B^*B \ge I$. The proof for $BB^* \ge I$ follows a similar argument.
\end{proof}

\paragraph{Example.}
Let $B = \begin{pmatrix} 1 & 2 \\ 3 & 4 \end{pmatrix}$. The permanent is $\per(B) = 1 \cdot 4 + 2 \cdot 3 = 10$.
The permanental inverse $B^*$ is:
\begin{equation*}
	B^* = \frac{1}{10} \begin{pmatrix} \per(B_{1,1}) & \per(B_{2,1}) \\ \per(B_{1,2}) & \per(B_{2,2}) \end{pmatrix} = \frac{1}{10} \begin{pmatrix} 4 & 2 \\ 3 & 1 \end{pmatrix}.
\end{equation*}
Now, let's compute the products:
\begin{align*}
	B^*B &= \frac{1}{10} \begin{pmatrix} 4 & 2 \\ 3 & 1 \end{pmatrix} \begin{pmatrix} 1 & 2 \\ 3 & 4 \end{pmatrix} = \frac{1}{10} \begin{pmatrix} 10 & 16 \\ 6 & 10 \end{pmatrix} = \begin{pmatrix} 1 & 1.6 \\ 0.6 & 1 \end{pmatrix}, \quad \text{and} \\
	BB^* &= \frac{1}{10} \begin{pmatrix} 1 & 2 \\ 3 & 4 \end{pmatrix} \begin{pmatrix} 4 & 2 \\ 3 & 1 \end{pmatrix} = \frac{1}{10} \begin{pmatrix} 10 & 4 \\ 24 & 10 \end{pmatrix} = \begin{pmatrix} 1 & 0.4 \\ 2.4 & 1 \end{pmatrix}.
\end{align*}
As demonstrated, both products \( B^*B \) and \( BB^* \) satisfy the entrywise inequality with respect to the identity matrix \( I \), but they are not necessarily equal: \( B^*B \neq BB^* \) in general. A straightforward consequence of Claim \ref{claim:perm-inverse-basic} is that
\begin{equation*}
    \per(B^*B) \geq 1 \quad \text{and} \quad \per(BB^*) \geq 1,
\end{equation*}
since both \( B^*B \) and \( BB^* \) dominate \( I \) entrywise. However, a sharper inequality is established in Theorem \ref{thm:inverse-property-perm} (see next section), which implies that
\begin{align}
	\label{eqn:perm-perm-inv-prod}
\per(B^*) \cdot \per(B) \geq 1.
\end{align}
This is strictly stronger than the two previous inequalities, because
\begin{equation*}
    \per(B^*B), \, \per(BB^*) \; \geq \; \per(B^*) \cdot \per(B).
\end{equation*}
The final inequality follows from the fact that the permanent is super-multiplicative on non-negative square matrices; that is, for any such matrices \( C, D \), we have \( \per(CD) \geq \per(C) \cdot \per(D) \).
\subsection{An Inequality for the Permanental Inverse}
The following inequality describes a relation between the ``permanental minors'' of a non-negative matrix and its permanental inverse: 
\begin{theorem}
		\label{thm:inverse-property-perm}
	If $B^*$ is the permanental inverse of a non-negative matrix $B$, then for any index sets $S$ and $T$, the following inequality holds:
	\begin{equation*}
		\frac{\per(B(-S,-T))}{\per(B)} \le \per(B^*(T,S)).
	\end{equation*}
	For context, the equivalent identity for determinants is an equality:
	\begin{equation*}
		\frac{\det(B(-S,-T))}{\det(B)} = \det(B^{-1}(T,S))
	\end{equation*}
\end{theorem}
We now examine some illustrative special cases of this inequality:
\begin{itemize}
	\item \textbf{Case 1:} \( S = T = [n] \). Then  \( B(-S, -T) \) is empty and \( B^*(T, S) =B^*\), and the inequality reads:
	\begin{equation*}
	    \frac{1}{\per(B)} \;\le\; \per(B^*) \quad \Longrightarrow \quad \per(B^*) \cdot \per(B) \;\ge\; 1\,.
	\end{equation*}
	This is the inequality from qquation \ref{eqn:perm-perm-inv-prod} before.
	\item \textbf{Case 2:} \( S = \{i\} \), \( T = \{j\} \). Then \( B(-S,-T) =B_{i,j}\) is the \( (n-1)\times(n-1) \) matrix obtained by deleting row \( i \) and column \( j \), and the inequality becomes:
	\begin{equation*}
	    \frac{\per(B_{i,j})}{\per(B)} \;\le\; \per((B^*)_{j,i})=(B^*)_{j,i}\,,
	\end{equation*}
	which holds with equality by definition of the permanental inverse.
\end{itemize}

These special cases highlight the role of the permanental inverse as a natural upper bound on normalized minors of $B $. While the determinantal analogue yields exact identities due to multiplicativity, the permanent lacks such algebraic structure.
\subsection{Schur's Formula for Permanents}
\label{sec:schur-permanent}
In contrast to Schur's determinant formula (which yields an exact equality), the analogous relationship for matrix permanents turns out to be an inequality. We formalize this below.
\PermanentInequality*
\noindent\textbf{Proof strategy.} The proof of Theorem~\ref{thm:permanent-inequality} will use induction on $k$ (the number of columns in $Y$). We first establish two auxiliary results: a formula for the permanent of a rank-1 block update (Observation~\ref{obs:rank-1-perm-update}) and a technical inequality (Lemma~\ref{lem:row-uncrossing-perm}) referred to as the \emph{row-uncrossing lemma}. After proving these, we proceed to the inductive step for the general case.

\begin{observation}[Rank-1 Update Formula]	\label{obs:rank-1-perm-update}
	For any block matrix of the form $\displaystyle\begin{pmatrix} B & y \\[3pt] x^\top & w \end{pmatrix}$, where $B$ is a $d\times d$ matrix, $x$ and $y$ are column vectors of length $d$, and $w$ is a scalar, the permanent can be expanded as
	\begin{equation*}
	    \per\begin{pmatrix} 
		B & y\\[3pt] 
		x^\top & w 
	\end{pmatrix}
	\;=\; \per(B)\,\cdot\,\Big(w \,+\, x^\top B^*\, y\Big)\,.
	\end{equation*}
\end{observation}
\begin{proof}
	We let $B^{(i \leftarrow y)}$ denote the $d\times d$ matrix obtained from $B$ by replacing its $i$-th column with the vector $y$. Expanding the permanent of the block matrix along its last row gives:
	\begin{equation*}
	    \per\begin{pmatrix} B & y \\ x^{\top} & w \end{pmatrix}
	\;=\; w \cdot \per(B) \;+\; \sum_{i=1}^d x_i \cdot \per\!\big( B^{(i\leftarrow y)} \big)\,. 
	\end{equation*}
	Here the first term $w\cdot\per(B)$ corresponds to choosing the entry $w$ in the last row, while each summand $x_i \cdot \per(B^{(i\leftarrow y)})$ corresponds to choosing the entry $x_i$ from the last row and then taking all permutations in the remainder of the matrix that involve one element from the inserted column $y$. In particular, if $y_j$ (the $j$-th entry of $y$) is used from that inserted column, it contributes a factor $x_i y_j$ and leaves a $(d-1)\times(d-1)$ submatrix $B_{j,i}$ (obtained by removing the $j$-th row and $i$-th column from $B$) for the rest of the permutation. Summing over all choices of $j$ for each $i$, we can rewrite the above as 
	\begin{equation*}
	    \per\begin{pmatrix} B & y \\ x^{\top} & w \end{pmatrix}
	\;=\; w \cdot \per(B) \;+\; \sum_{i=1}^d \sum_{j=1}^d x_i\,y_j\,\per(B_{j,i})\,.
	\end{equation*}
	Now, factor $\per(B)$ out of the summation. By the definition of $B^*$, we have $\frac{\per(B_{j,i})}{\per(B)} = (B^*)_{i,j}$. Thus,
	\begin{equation*}
	    \per\begin{pmatrix} B & y \\ x^{\top} & w \end{pmatrix}
	\;=\; \per(B)\Big( w + \sum_{i,j=1}^d x_i y_j\,\frac{\per(B_{j,i})}{\per(B)} \Big)
	\;=\; \per(B)\,\Big( w + x^\top B^*\,y \Big)\!,
	\end{equation*}
	which confirms the formula.
\end{proof}

\begin{lemma}[Row-Uncrossing Inequality]
	\label{lem:row-uncrossing-perm}
	Let
	\begin{equation*}
	    M = \begin{bmatrix}
		B & Y \\
		X^{\top} & W
	\end{bmatrix},
	\end{equation*}
	where 
	\(
	B \in \RR^{d \times d}_{\ge 0},\;
	X, Y \in \RR^{d \times k}_{\ge 0},\;
	W \in \RR^{k \times k}_{\ge 0}
	\)
	with \(d \ge 0\) and \(k \ge 1\).
	Then, for any fixed \( i^* \in [k] \), the following inequality holds:
	\begin{equation}
		\label{eqn:row-uncrossing-perm}
		\per(M) \cdot \per(B)
		\;\le\;
		\sum_{j=1}^k 
		\per\begin{bmatrix}
			B & Y_{., -j} \\
			X_{-i^*, .}^{\top} & W_{i^*, j}
		\end{bmatrix}
		\cdot
		\per\begin{bmatrix}
			B & y_j \\
			x_{i^*}^{\top} & w_{i^*, j}
		\end{bmatrix}.
	\end{equation}
	Here, \( Y_{.,\, -j} \) denotes the matrix \( Y \) with its \( j \)-th column removed, and \( X_{-i^*, .}^{\top} \) denotes \( X^{\top} \) with its \( i^* \)-th row removed (equivalently, removing the \( i^* \)-th column of \( X \) before transposing). Likewise, \( W_{i^*,\,j} \) is the submatrix of \( W \) obtained by deleting the \( i^* \)-th row and \( j \)-th column.
\end{lemma}

\begin{proof}
The proof proceeds in two main stages:
\begin{enumerate}
	\item We first establish the inequality in the special case where \( W = 0 \). This is done by explicitly expanding the permanents and applying an inductive argument on the dimension \( d \).
	
	\item We then extend the result to arbitrary non-negative matrices \( W \). To do this, we define a function representing the difference between the two sides of the inequality. From the first step, we know that this function is non-negative when \( W = 0 \). We observe that the function is multilinear in the entries of \( W \), and that all its partial derivatives are non-negative. Each partial derivative corresponds to an instance of the same inequality, but for a smaller value of \( k \), allowing us to invoke the inductive hypothesis. These observations imply that the difference function remains non-negative for all non-negative \( W \), thereby completing the proof.
\end{enumerate}
\textbf{The base cases $d=0$ or $k=1$:} For the base case when $d=0$, the inequality in \eqref{eqn:row-uncrossing-perm} is 
\begin{align*}
	\per(W)\leq \sumL_{1\leq j\leq k}\per(W_{i^*,j})\cdot w_{i^*,j} ,
\end{align*}
which holds with equality as this is precisely the Laplace expansion of $\perm(W)$ at row $i^*$. For the base case when $k=1$, the inequality in \eqref{eqn:row-uncrossing-perm} is 
\begin{align*}
	\per\begin{bmatrix}
		B & y \\
		x^{\top} & w
	\end{bmatrix}\cdot \per(B) \leq  
	\per(B)\cdot \per\begin{bmatrix}
		B & y \\
		x^{\top} & w
	\end{bmatrix} 
\end{align*}
which is trivially true with equality.  So assume that $d\geq 1$ and $k\geq 2$ in any case from here on. 

\textbf{The Special Case $W=0$:} We use induction on $d$.  For $d\geq 1$, observe that when $d<k$, the LHS term is equal to zero when $W=0$. Because a term in the permanent of \(M\) (with \(W=0\)) is non-zero only if it selects \(d-k\) elements from \(B\) and \(k\) elements from each of \(X\) and \(Y\). Assume that $d\geq k$ and let $\cF$ be the set of one-one functions mapping from $[k]$ to $[d]$. Expanding the LHS of \eqref{eqn:row-uncrossing-perm} gives 
\begin{align*}
	\per\begin{bmatrix}
		B & Y \\
		X^{\top} & 0
	\end{bmatrix}\cdot \per(B)= \per(B)\cdot \sumL_{f,g\in \cF} \prod_{i\in [k]}x_{f(i),i}\cdot \prod_{j\in [k]}y_{g(j),j}\cdot \per B(-\textnormal{img}_g,-\textnormal{img}_f).    
\end{align*}
In order to similarly expand the RHS, define $\cF_t$ for $t\in [k]$ as the set of functions mapping $[k]$ to $ [d]$ such that $f\in \cF_t$ is one-one when restricted to $[k]\del\{t\}$. Expanding the RHS of \eqref{eqn:row-uncrossing-perm} gives
\begin{align*}
	&\sumL_{1\leq t\leq k} 
	\per\begin{bmatrix}
		B & Y_{.,-t} \\
		X_{-i^*,.}^{\top} & 0
	\end{bmatrix}\cdot \per\begin{bmatrix}
		B & y_{t} \\
		x_{i^*}^{\top} & 0
	\end{bmatrix}= \\&\sumL_{1\leq t \leq k}\sumL_{f'\in \cF_{i^*}, g'\in \cF_{t}}\prod_{i\in [k]}x_{f'(i),i}\cdot \prod_{j\in [k]}y_{g'(j),j}\per[B(-\textnormal{img}_{g',[k]-t} ,-\textnormal{img}_{f',[k]-i^*})]\cdot\per(B_{g'(t),f'(i^*)}).
	\intertext{Since $\cF\subseteq \cF_t$ for any $t$, we can obtain a lower bound by only summing over $f',g'\in \cF$. By exchanging the summations after this step gives}
	&\geq \sumL_{f',g'\in \cF}\prod_{i\in [k]}x_{f'(i),i}\cdot \prod_{j\in [k]}y_{g'(j),j} \sumL_{1\leq t\leq k}\per[B(-\textnormal{img}_{g',[k]-t},-\textnormal{img}_{f',[k]-i^*} )]\cdot\per(B_{g'(t),f'(i^*)}).
\end{align*}
	It is sufficient to show that 
\begin{align}
	\label{eqn: smallerform-row-uncrossing-perm}
	\per(B)\cdot \per B(-\textnormal{img}_g,-\textnormal{img}_f) \leq \sumL_{1\leq t\leq k}\per[B(-\textnormal{img}_{g,[k]-t},-\textnormal{img}_{f,[k]-i^*} )]\cdot\per(B_{g(t),f(i^*)})
\end{align}
for any $f,g\in \cF, \, i^*\in [k]$. Equation \eqref{eqn: smallerform-row-uncrossing-perm} is in the form of \eqref{eqn:row-uncrossing-perm} with the substitution 
\begin{align*}
	B,X,Y,i^*,d,k \leftarrow B(-\textnormal{img}_g,-\textnormal{img}_f)^{\top}, B(-\textnormal{img}_g,\textnormal{img}_f), B(\textnormal{img}_g,-\textnormal{img}_f)^{\top} , f(i^*), d-k, k. 
\end{align*}
We can conclude this case using inductive hypothesis. 

\textbf{Extension to $W\geq 0$:} 	It remains to prove \eqref{eqn:row-uncrossing-perm} for general $W$ given that we have a proof for the case when $W=0$. The first step is to observe that both the LHS and RHS of \eqref{eqn:row-uncrossing-perm} are multi-linear functions with respect to the $w_{i,j}$ variables. For fixed $B, X, Y, i^*$, consider the function 
$h(W):=  \sumL_{1\leq j\leq k} 
\per\begin{bmatrix}
	B & Y_{.,-j} \\
	X_{-i^*,.}^{\top} & W_{i^*,j}
\end{bmatrix}\cdot \per\begin{bmatrix}
	B & y_{j} \\
	x_{i^*}^{\top} & w_{i^*,j}
\end{bmatrix}-\per\begin{bmatrix}
	B & Y \\
	X^{\top} & W
\end{bmatrix}\cdot \per(B) $. Since we know that $h(0)\geq 0$, it suffices to show that $\frac{\partial h(W)}{\partial w_{\alpha, \beta}}\geq 0$ for every $\alpha,\beta\in [k]$. Then that would imply $h
(W)\geq h(0)\geq 0$. 

We proceed by induction on $k$. For $\alpha, \beta \in [k], \, \alpha\neq i^*$, the derivative $\frac{\partial h(W)}{\partial w_{\alpha, \beta}}$ is equal to 
\begin{align*}
	\sumL_{j\in [k]\backslash \{\beta\}}\per\begin{bmatrix}
		B & Y([d],-\{j,\beta\}) \\
		X([d], -\{i^*, \alpha\})^{\top} & W(-\{i^*,\alpha\}, -\{j,\beta\})
	\end{bmatrix}\cdot \per\begin{bmatrix}
		B & y_{j} \\
		x_{i^*}^{\top} & w_{i^*,j}
	\end{bmatrix} -\per\begin{bmatrix}
		B & Y_{.,-\beta} \\
		X_{-\alpha,.}^{\top} & W_{\alpha,\beta}
	\end{bmatrix}\cdot \per(B) 
\end{align*}
which is non-negative using inductive hypothesis. The substitution being 
\begin{align*}
	B,X,Y,i^*,d,k \leftarrow B, X_{.,-\alpha}, Y_{.,-\beta},i^*,d,k-1.
\end{align*}
For $\alpha=i^*, \beta\in [k]$, the derivative $\frac{\partial h(W)}{\partial w_{i^*, \beta}}$ is equal to
\begin{align*}
	\per\begin{bmatrix}
		B & Y_{.,-\beta} \\
		X_{-i^*,.}^{\top} & W_{i^*,\beta}
	\end{bmatrix} \cdot\per(B)- \per\begin{bmatrix}
		B & Y_{.,-\beta} \\
		X_{-i^*,.}^{\top} & W_{i^*,\beta}
	\end{bmatrix} \cdot\per(B)=0.
\end{align*}
This finishes the proof of the lemma.
	\end{proof}
\begin{proof}[Proof of Theorem \ref{thm:permanent-inequality}]
We prove by induction on $k$. Equation \eqref{eqn:perm-inv-prop} holds with equality for $k=1$ using Observation \ref{obs:rank-1-perm-update}. For $k\geq 2$, using Lemma \ref{lem:row-uncrossing-perm}, we have	
	\begin{align*}	
		\per\begin{bmatrix}		
			B & Y \\		
			X^{\top} & W			
		\end{bmatrix}\cdot \per(B) &\leq \sumL_{1\leq j\leq k}	
		\per\begin{bmatrix}			
			B & Y_{.,-j} \\			
			X_{-1,.}^{\top} & W_{1,j}			
		\end{bmatrix}\cdot \per\begin{bmatrix}	
			B & y_{j} \\		
			x_{1}^{\top} & w_{1,j}		
		\end{bmatrix}.
		\intertext{			
			Using inductive hypothesis, we have		
		}		
		\per\begin{bmatrix}		
			B & Y_{.,-j} \\		
			X_{-1,.}^{\top} & W_{1,j}		
		\end{bmatrix}&\leq \per(B)\cdot \per(W_{1,j}+X_{-1,.}^{\top} B^*Y_{.,-j}).	
		\intertext{substituting this gives}	
		\per\begin{bmatrix}		
			B & Y \\		
			X^{\top} & W		
		\end{bmatrix}\cdot \per(B) &\leq \per(B)\sumL_{1\leq j\leq k}\per(W_{1,j}+X_{-1,.}^{\top} B^*Y_{.,-j})\cdot \per\begin{bmatrix}		
			B & y_{j} \\		
			x_{1}^{\top} & w_{1,j}.		
		\end{bmatrix}	
		\intertext{Using Observation \ref{obs:rank-1-perm-update} gives}	
		&\leq \per(B)^2\sumL_{1\leq j\leq k}\per(W_{1,j}+X_{-1,.}^{\top} B^*Y_{.,j})\cdot (w_{1,j}+x_1^{\top}B^*y_j)\\	
		&=\per(B)^2\cdot \per(W+X^{\top}B^*Y).	
	\end{align*}	
	concluding the proof that
		\begin{align*}	
		\per\begin{bmatrix}			
			B & Y \\	
			X^{\top} & W		
		\end{bmatrix} \leq \per(B)\cdot \per(W+X^{\top}B^*Y).	
	\end{align*}	
\end{proof}
\begin{proof}[Proof of Theorem \ref{thm:inverse-property-perm}]
    Setting the columns of $X$ and $Y$ as the standard basis vectors corresponding to indices given by sets $S,T$ respectively in Theorem \ref{thm:permanent-inequality} gives Theorem \ref{thm:inverse-property-perm}. 
\end{proof}
In the following subsection, we mention two special cases of Theorem \ref{thm:permanent-inequality} when $k$ or $d$ is small. 
\subsubsection{Two permanent inequalities}
We state two special cases of Theorem \ref{thm:permanent-inequality} in Lemma \ref{lem:permanent-uncrossing} and Lemma \ref{lem:perm-ineq-2}.

\begin{lemma}
	\label{lem:permanent-uncrossing}
   For any matrix $B\in \RR^{d\times d}_{\geq 0}$, vectors $x_i,y_i \in \RR^{d\times 1}_{\geq 0}$, and scalars $w_{i,j}\geq 0$ with $i,j \in \{1,2\}$, the following inequality holds true:
	\begin{align}
		\label{eqn:ineq-1}
		\per\begin{bmatrix}
			B & y_1 & y_2 \\
			x_1^{\top}& w_{1,1} & w_{1,2}\\
			x_2^{\top}& w_{2,1} & w_{2,2}
		\end{bmatrix}\cdot \per(B)\leq \per&\begin{bmatrix}
			B & y_1 \\
			x_1^{\top}& w_{1,1}
		\end{bmatrix}  \cdot\per\begin{bmatrix}
			B & y_2 \\
			x_2^{\top}& w_{2,2}
		\end{bmatrix}\notag \\
        &+ \per\begin{bmatrix}
			B & y_2 \\
			x_1^{\top}& w_{1,2}
		\end{bmatrix}\cdot\per\begin{bmatrix}
			B & y_1 \\
			x_2^{\top}& w_{2,1}
		\end{bmatrix}.
	\end{align}	
\end{lemma}
The above lemma is a permanental version of Lemma \ref{lem:schur-determinant}.
\begin{lemma} \label{lem:perm-ineq-2}
	For a matrix $W\in \RR^{k\times k}_{\geq 0}$, vectors $x,y\in \RR^k_{\geq 0}$, and scalar $b\neq 0$, let $C\in \RR^{k\times k}_{\geq 0}$ be the matrix defined by $c_{i,j}:= \per\begin{bmatrix}
		b & y_j\\
		x_i& w_{i,j}
	\end{bmatrix}\cdot b^{-1}=w_{i,j}+x_iy_j/b$. The following inequality holds true:
	\begin{align}
		\label{eqn:ineq-2}
		\per\begin{bmatrix}
			b & y^{\top} \\
			x & W 
		\end{bmatrix}\cdot b^{-1} \leq \per(C).
	\end{align}
\end{lemma}

\begin{remark}
    We prove an analogue of Lemma \ref{lem:perm-ineq-2} for PSD matrices in Section \ref{subsec:psd}, as it is crucial for the permanent process.
\end{remark}
\section{The Permanent Process}
\label{sec:perm-process}
Let $A^{(t)}$ denote the matrix at the beginning of step $t$ of Algorithm \ref{alg:perm-process}. The state of the matrix entries after the end of step $t$ for some $1\leq t\leq n-1$ is given by 
\begin{align}
	\label{eqn:perm-process}
	a_{i,j}^{(t+1)}=\begin{cases} a_{i,j}^{(t)}+\frac{a_{i,t}^{(t)}a_{t,j}^{(t)}}{a_{t,t}^{(t)}}  , &\text{for $i,j\geq t+1$}\\
		a_{i,j}^{(t)}, &\text{otherwise.}
	\end{cases}
\end{align}
It is difficult to obtain a clean closed-form solution for the entries of $A^{(t)}$ with respect to the entries of $A$ during the permanent process, like we had for the Gaussian elimination process in Theorem \ref{thm:Gaussian-determinant}. However, something weaker can be proven that is strong enough to extend Corollary \ref{cor:Gauss_elim-invar}. 
\begin{theorem}
	\label{thm:perm_process-invar}
	Let $A$ be a real non-negative or a real PSD matrix. The diagonal entries of the matrix $A^{(n)}$ can be lower bounded by 
	\begin{align*}
		a_{t,t}^{(n)}=a_{t,t}^{(t)}\geq \frac{\per(A^{(t)}(-[t-1], -[t-1]))}{\per(A^{(t+1)}(-[t], -[t]))}.
	\end{align*}
	for $1\leq t\leq n$.
\end{theorem}
\begin{proof}
	Observe that it is sufficient to prove just the case for $t=1$ because the second step of the permanent process essentially applies the first step of the permanent process on the sub-matrix $A^{(2)}(-\{1\}, -\{1\})$. What we want to show is that 
	\begin{align}
		\label{eqn:perm-ineq}
		\per(A^{(2)}(-\{1\}, -\{1\}))  \geq \frac{\per(A)}{a_{1,1}}.
	\end{align}
	This  follows directly from Lemma \ref{lem:perm-ineq-2} for non-negative matrices and Lemma \ref{lem:psd-ineq} for PSD matrices. It is crucial to observe that the non-negativity and PSD property remain satisfied in the matrix $A^{(2)}(-\{1\}, -\{1\})$ which allows us to use \eqref{eqn:perm-ineq} recursively. 
\end{proof}

We now have all the ingredients required to prove Theorem \ref{thm:perm-process-invar}, which we restate for the readers' convenience.
\permprocess*

\begin{proof}
	Multiplying the lower bounds for $a_{i,i}^{(n)}$ over $i \in \{1,\ldots, n\}$ from Theorem \ref{thm:perm-process-invar} gives the required lower bound.
\end{proof}

Theorem \ref{thm:perm-process-invar} provides an algorithmic upper bound for the permanent of any non-negative or PSD matrix. In fact, any theoretical upper bound to the product of the diagonal entries of $A$ after $n$ steps of the permanent process to $A$ can be used as an upper bound for $\per(A)$. 

\subsection{Postitive Semidefinite Matirces} \label{subsec:psd}
\begin{lemma}
\label{lem:psd-ineq}
    Let $A = \begin{bmatrix}
        B & x \\
        x^\top & a
    \end{bmatrix}$ be an $n \times n$ real PSD matrix with $a \in \mathbb{R}$, $x \in \mathbb{R}^{n-1}$ and $B \in \mathbb{R}^{(n-1)\times (n-1)}$. The permanent of $A$ satisfies  
    \begin{equation*}
       \per(A) =  \per\begin{bmatrix}
        B & x \\
        x^\top & a
    \end{bmatrix} \leq a \cdot \per\left( B + \frac{xx^\top}{a} \right).
    \end{equation*}
\end{lemma}
\begin{proof}
We can re-write the R.H.S. of the above inequality as
\begin{equation*}
     a \cdot \per\left( B + \frac{xx^\top}{a} \right) = a^{-(n-2)} \per\left( a \cdot B + xx^\top \right).
\end{equation*}
Define $\alpha_0, \ldots, \alpha_n$ as the coefficients of the powers of $a$ in the expansion of $\per( a \cdot B + xx^\top )$, i.e.,
\begin{equation*}
    \per\left( a \cdot B + xx^\top \right) := \sum_{\ell = 0}^n a^\ell \alpha_{n-\ell}.
\end{equation*}
Note that $\alpha_0 = \per(B)$ and $ \alpha_1 = \sum_{i,j \in [n]} x_i x_j \per(B_{i,j})$. 
Therefore, 
\begin{equation*}
    a^2 \alpha_0 + a \alpha_1 = a \cdot \left( a \cdot \per(B) + \sum_{i,j \in [n]} x_i x_j \per(B_{i,j}))\right) = \per\begin{bmatrix}
        B & x \\
        x^\top & a
    \end{bmatrix}.
\end{equation*}
So, to complete the proof, it suffices to show that $\alpha_\ell \geq 0$ for all $\ell \in \{2, \ldots, n\}$.

Let $A = V^\top V$. Then $a = v_n^\top v_n$ and $x = [v_1 \; v_2 \; \ldots \;v_{n-1}]^\top v_n$.
    Using the formula in Lemma \ref{lem:psd-perm}, we have
    \begin{equation*}
        \per(B) = \frac{1}{(n-1)!} \norm{\sum_{\sigma \in \mathcal{S}_{n-1}} v_{\sigma(1)} \otimes v_{\sigma(2)} \ldots \otimes v_{\sigma(n-1)} }^2.
    \end{equation*}

Therefore, the permanent of $B$ can be expanded as 
\begin{equation*}
        \per(B) = \frac{1}{(n-1)!} \sum_{\sigma, \pi \in \mathcal{S}_{n-1}} \prod_{i=1}^{n-1} (v_{\sigma(i)}^\top v_{\pi(i)}). 
\end{equation*}

Observe that each term in $\alpha_{n-\ell}$ involves selecting a set of $\ell$ indices $S$ from $\{1, \ldots, n\}$,  and replacing $v_{\sigma(i)}^\top v_{\pi(i)}$ with $(v_{\sigma(i)}^\top v_n) \cdot (v_n^\top v_{\pi(i)})$ for each $i \in S$ in the above expression.
Therefore, 
\begin{align*}
    \alpha_{n-\ell} 
    &= \frac{1}{(n-1)!}   \sum_{\sigma, \pi \in \mathcal{S}_{n-1}} \sum_{S \in \binom{[n-1]}{[\ell]}}  \prod_{i \notin S} (v_{\sigma(i)}^\top v_{\pi(i)})  \cdot \prod_{i \in S} (v_{\sigma(i)}^\top v_n) \cdot (v_n^\top v_{\pi(i)}).
\end{align*}
Rearranging the terms gives,
    \begin{equation*}
        \alpha_{n-\ell} = \frac{\binom{n-1}{\ell}}{(n-1)!} \norm{\sum_{\sigma \in \mathcal{S}_{n-1}} \left( v_n^\top v_{\sigma(1)} \cdot v_n^\top v_{\sigma(2)} \cdot\ldots\cdot v_n^\top v_{\sigma(\ell)}\right)  \cdot  v_{\sigma(\ell+1)} \otimes v_{\sigma(\ell+2)} \ldots  \otimes v_{\sigma(n-1)} }^2.
    \end{equation*}
    Here the $\binom{n-1}{\ell}$ factor accounts for selecting the first $\ell$ indices to be in the set $S$.
    
    Therefore, $\alpha_{n-\ell} \geq 0$ for all $\ell \in \{0, \ldots, n\}$ and this completes the proof.
\end{proof}

\subsection{Recursive Upper Bounds}
Expanding the recursive definition of the permanent process from \eqref{eqn:perm-process}, we obtain:
\begin{align}
	\label{eqn:permanent-process-expanded}
	a_{i,j}^{(t)} = a_{i,j}^{(1)} + \sum_{s=1}^{\min(t,j)} \frac{a_{i,s}^{(s)} a_{s,j}^{(s)}}{a_{s,s}^{(s)}}, \quad \forall t.
\end{align}

This recurrence suggests focusing on entries of the form \( a_{i,j}^{(\min(i,j))} \), since they can be expressed in terms of similar entries with smaller indices. Substituting \( t = \min(i,j) \) in \eqref{eqn:permanent-process-expanded} yields:
\begin{align}
	\label{eqn:permanent-process-expanded-special}
	a_{i,j}^{(\min(i,j))} 
	&= a_{i,j}^{(1)} + \sum_{s=1}^{\min(i,j)} \frac{a_{i,s}^{(s)} a_{s,j}^{(s)}}{a_{s,s}^{(s)}}.
\end{align}
Observe that in each term \( a_{i,s}^{(s)} \) and \( a_{s,j}^{(s)} \), the index \( s \) satisfies \( s = \min(i,s) = \min(s,j) \) since \( s < \min(i,j) \).

Define the shorthand:
\begin{equation*}
    u_{i,j} := a_{i,j}^{(\min(i,j))}.
\end{equation*}
Substituting this into \eqref{eqn:permanent-process-expanded-special} gives the recurrence:
\begin{align}
	\label{eqn:permanent-process-expanded-upperbound-dp}
	u_{i,j} = a_{i,j} + \sum_{s=1}^{\min(i,j)} \frac{u_{i,s} u_{s,j}}{u_{s,s}}.
\end{align}

\begin{theorem}
	\label{thm:perm-upperbound-recursion}
	Let \( A \in \RR^{n \times n}_{\geq 0} \) be a non-negative matrix. If a matrix \( B \in \RR^{n \times n}_{\geq 0} \) satisfies:
	\begin{align}
		\label{eqn:perm-upperbound-recursion}
		a_{i,j} + \sum_{s=1}^{\min(i,j)} \frac{b_{i,s} b_{s,j}}{a_{s,s}} \le b_{i,j},
	\end{align}
	then \( u_{i,j} \le b_{i,j} \) for all \( i,j \in [n] \). Moreover, this conclusion remains valid even if the inequality in \eqref{eqn:perm-upperbound-recursion} holds with equality.
\end{theorem}

\begin{proof}
	We prove the claim by induction on \( t = \min(i,j) \).
	
	\textbf{Base case:} If \( t = 1 \), then \( u_{i,j} = a_{i,j} \le b_{i,j} \) directly from the assumption.
	
	\textbf{Inductive step:} Suppose the claim holds for all pairs \( (i,j) \) with \( \min(i,j) < t \). Consider \( \min(i,j) = t \ge 2 \). Using the recurrence in \eqref{eqn:permanent-process-expanded-upperbound-dp}, we have:
	\begin{align*}
		u_{i,j} &= a_{i,j} + \sum_{s=1}^t \frac{u_{i,s} u_{s,j}}{u_{s,s}} \\
		&\le a_{i,j} + \sum_{s=1}^t \frac{u_{i,s} u_{s,j}}{a_{s,s}} \\
		&\le a_{i,j} + \sum_{s=1}^t \frac{b_{i,s} b_{s,j}}{a_{s,s}} \le b_{i,j},
	\end{align*}
	where the second inequality uses the inductive hypothesis \( u_{i,s}, u_{s,j} \le b_{i,s}, b_{s,j} \), and the last step uses the assumption in \eqref{eqn:perm-upperbound-recursion}.
	
	The same argument applies if the inequality in \eqref{eqn:perm-upperbound-recursion} holds with equality. Thus, the result holds in both the inequality and equality cases.
\end{proof}

\begin{corollary}
	\label{cor: permanent-diag-prod-ub}
	Let \( A \in \RR^{n \times n}_{\geq 0} \), and let \( B \in \RR^{n \times n}_{\geq 0} \) satisfy \eqref{eqn:perm-upperbound-recursion}. Then:
	\begin{equation*}
	    \per(A) \le \prod_{i=1}^n b_{i,i}.
	\end{equation*}
\end{corollary}

\begin{proof}
	From Theorem~\ref{thm:perm-upperbound-recursion}, we have \( u_{i,i} \le b_{i,i} \) for all \( i \). By Theorem~\ref{thm:perm_process-invar}, we know that \( \per(A) \le \prod_{i=1}^n u_{i,i} \). Combining both inequalities yields the desired bound.
\end{proof}

\subsection{Theoretical Upper Bounds for Structured Matrices}
\label{sec:theoretical-ub}
We now illustrate how the recursive upper bound framework can be used to derive explicit upper bounds for permanents of structured matrices. 

Recall from Theorem~\ref{thm:perm-upperbound-recursion} that if two non-negative matrices \( A, B \in \RR^{n \times n}_{\geq 0} \) satisfy:
\begin{align}
	\label{eqn:a-to-b}
	a_{i,j} + \sum_{ s = 1}^{\min(i,j)} \frac{b_{i,s} b_{s,j}}{a_{s,s}} = b_{i,j} \quad \text{for all } i,j \in [n],
\end{align}
then:
\begin{equation*}
    \per(A) \le \prod_{i=1}^n b_{i,i}.
\end{equation*}

This identity suggests a simple approach: to upper bound \( \per(A) \), we can construct a matrix \( B \) satisfying \eqref{eqn:a-to-b} and then bound the entries of \( B \) in terms of those of \( A \).

\applicationbound*
\begin{proof}
	We aim to show that the matrix \( B \) defined via \eqref{eqn:a-to-b} satisfies \( b_{i,j} \le (1+\varepsilon) a_{i,j} \). The result will then follow from Corollary~\ref{cor: permanent-diag-prod-ub}.
	
	We proceed by induction on \( \min(i,j) \). The base case \( \min(i,j) = 1 \) is immediate since \( b_{i,j} = a_{i,j} \) in this case.
	
	For \( i,j \ge 2 \), using the definition \eqref{eqn:a-to-b} and the inductive hypothesis, we have:
	\begin{align*}
		b_{i,j} &= a_{i,j} + \sum_{ s = 1}^{\min(i,j)} \frac{b_{i,s} b_{s,j}}{a_{s,s}} \\
		&\le a_{i,j} + (1+\varepsilon)^2 \sum_{ s = 1}^{\min(i,j)} \frac{a_{i,s} a_{s,j}}{a_{s,s}} \\
		&\le a_{i,j} + \varepsilon a_{i,j} = (1+\varepsilon) a_{i,j},
	\end{align*}
	where the second step uses the inductive assumption \( b_{i,s}, b_{s,j} \le (1+\varepsilon) a_{i,s}, a_{s,j} \), and the third uses the assumption in the lemma.
	
	This completes the inductive step and hence the proof.
\end{proof}

\subsection{A Concrete Example}
Let $f_n: \mathbb{N} \rightarrow \mathbb{R}_{\geq 0}$ be a family of functions and consider a family of matrices $\{ A_n\}_{n \in \mathbb{N}}$ defined using $f$ as follows: $A_n$ is an $n \times n$ matrix with all diagonal entries $1$ and the off-diagonal entries satisfy 
\begin{equation*}
    (A_n)_{i,j} \leq f_n(i-j) .
    \end{equation*}

A natural question is to understand the growth rate of $\per(A_n)$, i.e., how large can $\per(A_n)$ be as a function of $n$: polynomial, exponential, or even larger? Since computing the permanent exactly becomes infeasible as $n$ grows, and the exact value is not the main concern, direct computation or classical methods are too expensive. The permanent process, however, provides an efficient way to obtain upper bounds on $\per(A_n)$. While the product-of-row-sums is one natural upper bound, it is often far larger than the permanent. In some cases, the permanent process yields sharper approximations, and in certain cases, it even captures the correct order of magnitude. We discuss a concrete example below.

Note that it suffices to consider matrices with all diagonal entries $1$ and the off-diagonal entries $(A_n)_{i,j} = f_n(i-j)$, as this only increases the permanent while making the calculations for the permanent process simpler. 

\paragraph{Exponential function.}

Consider a family of matrices $\{ A_n\}_{n \in \mathbb{N}}$, parameterized by a constant $c_n > 1$, defined as follows: $A_n$ is an $n \times n$ matrix with 
\begin{equation*}
    (A_n)_{i,j} = c_n^{-|i-j|}.
\end{equation*}

\begin{claim} \label{claim:exp-matrix}
    The entries of the matrix $A_n^{(n)}$ returned by Algorithm \ref{alg:perm-process} are given by 
\begin{equation*}
    (A^{(n)}_n)_{i,j} = A_{i,j} \cdot \left(1 + c_n^{-2} + 2 c_n^{-4} + \ldots + 2^{\min(i,j)-1} c_n^{-2\min(i,j)+2}\right).
\end{equation*}
\end{claim}
 For a proof of this claim, refer to Appendix \ref{sec:omitted-proofs}.
So, by Theorem \ref{thm:perm_process-invar}, we have
\begin{equation*}
    \per(A_n) = \prod_{i=1}^n a^{(i)}_{i,i} \leq \left(1 + c_n^{-2} + 2 c_n^{-4} + \ldots + 2^{n-1} c_n^{-2n+2}\right)^n
\end{equation*}
Since $1 + c_n^{-2} + 2 c_n^{-4} + \ldots + 2^{t-1} c_n^{-2n+2} < 1 + \frac{1}{c_n^2-2} $, we have 
\begin{equation*}
    \per(A_n) < \left(1+\frac{1}{c_n^2-2}\right)^n.
\end{equation*}

So, when $c_n = \Omega(\sqrt{n})$, $\per(A_n) \leq e$. However, the product-of-row-sums for such matrix is at least $(1 + \frac{1}{\sqrt{n}})^n$ and $(1 + \frac{1}{\sqrt{n}})^n \to \infty$ as $n\to \infty$.

\section{Numerical Stability of the Permanent Process} \label{sec:boundedness}
This section shows that the representation size of matrix entries remain polynomially bounded with respect to the input representation size during the execution of the permanent process. 
\begin{definition}[Bound function]
    For any given real $1\leq M$, define the function 
    \begin{align}
       B(n,k,t)=\gamma_n \cdot g(k,t), \textnormal{ where }\gamma_n = n!\cdot M^n \textnormal{ and } g(k,t)= M^k\cdot (M+1)^{t-1}. 
    \end{align}
\end{definition}
\begin{theorem} \label{thm:boundedness}
Let $A$ be an $n \times n$ matrix with $a_{i,i} = 1$ for all $i \in [n]$ and $0\leq a_{i,j} \le M$ for all $i,j$ for some $1\leq M$. For any $t \le \min(i, j)$, the entries of the matrix $A^{(t)}$ generated by the permanent process are bounded as follows:
\begin{align*}
    a_{i,j}^{(t)} &\le B(n,1,t) \\
    &=\gamma_n \cdot g(1,t)=n!\cdot M^{n+1}\cdot (M+1)^{t-1}.
\end{align*}
\end{theorem}
\begin{proof}
The proof proceeds by induction on $t$. The base case $t=1$ is trivial, as $a_{i,j}^{(1)} = a_{i,j} \le M$.

For the inductive step, assume the bounds hold for step $t-1$. The update rule for a diagonal entry is:
\begin{equation*}
    a^{(t)}_{i,i} = a^{(t-1)}_{i,i} + \frac{a^{(t-1)}_{i,t} a^{(t-1)}_{t,i}}{a^{(t-1)}_{t,t}}.
\end{equation*}
By the inductive hypothesis, the first term is bounded as $a^{(t-1)}_{i,i} \le B(n,1,t-1)$. The second term, which represents the update from the pivot, is a ratio of the form addressed by Lemma~\ref{lem:cycles_bound} below. Applying the lemma (with $|S|=2$, corresponding to the cycle $i \to t \to i$), we get:
\begin{equation*}
    \frac{a^{(t-1)}_{i,t} a^{(t-1)}_{t,i}}{a^{(t-1)}_{t,t}} \le B(n,2,t-1).
\end{equation*}
Combining these bounds, we get
\begin{equation*}
    a^{(t)}_{i,i} \le B(n,1, t-1) + B(n,2, t-1)= B(n,1,t).
\end{equation*}

This completes the induction for the diagonal entries. The proof for off-diagonal entries $a_{i,j}^{(t)}$ follows a similar structure, except for applying Lemma \ref{lem:cycles_bound} we switch the $i$-th column and the $j$-th column. This does not affect the entries of the matrix so far as $t \leq \min(i,j)$. 
\end{proof}

\subsection{Technical Lemmas}
\begin{observation} 
\label{obs:perm_ratio_bound}
Let $A$ be an $n \times n$ matrix with $a_{i,i} = 1$ for all $i \in [n]$ and $0\leq a_{i,j} \le M$ for all $i$ for some $1\leq M$. Then for every $ S\subseteq [n]$ and $i,j \in [n]\del S$
\begin{equation*}
    \frac{\per(A(S+i, S+j))}{\per(A(S,S))} \le \gamma_{|S|+1}.
\end{equation*}
\end{observation} 
\begin{proof}
Follows trivially from the fact that $\perm(A(S,S))\geq 1$ and $\perm(A(S+i, S+j))\leq M^{|S|+1}\cdot (|S|+1)!=\gamma_{|S|+1}$. 
\end{proof}

\begin{lemma}\label{lem:cycles_bound}
Let $A$ be an $n \times n$ matrix with $a_{i,i} = 1$ for all $i \in [n]$ and $0\leq a_{i,j} \le M$ for all $i,j$ for some $1\leq M$. Let $A^{(t)}$ be the matrix after $t-1$ steps of the permanent process. For any set of indices $S \subseteq \{t+1, \dots, n\}$ and a fixed element $i_0 \in S$, let $\mathcal{C}$ be the set of all cyclic permutations of elements in $S$. Then
\begin{equation*}
     \frac{\displaystyle \sum_{\sigma \in \mathcal{C}} \prod_{i \in S} a^{(t)}_{i,\sigma(i)}}
       {\per\!\big(A^{(t)}(S - i_0 ,\,S - i_0)\big)}
  \;\leq\; B(n,|S|,t).
\end{equation*}
\end{lemma}

\begin{proof}
    We will prove this by induction on $t$. For the base case, at $t = 1$, $A^{(1)} = A$ and by Observation \ref{obs:perm_ratio_bound},
    \begin{equation*}
          \frac{\sum_{\sigma \in \mathcal{C}} \prod_{i \in S} a_{i, \sigma(i)}}{\per(A_{S\backslash \{i_0\}, S \backslash \{i_0\}})} \leq \frac{\per(A_{S, S})}{\per(A_{S\backslash \{i_0\}, S \backslash \{i_0\}})} \leq |S|! \cdot M^{|S|} = \gamma_{|S|}\leq B(n,|S|,1).
    \end{equation*}

Without loss of generality, let $S=\{t+1, \ldots, t+k\}$ and let $i_0 = t+k$. 
    For $t > 1$, we have $a^{(t)}_{i,j} = a^{(t-1)}_{i,j} + \frac{a^{(t-1)}_{i,t} \cdot a^{(t-1)}_{t,j}}{a^{(t-1)}_{t,t}}$. 
    For ease of notation, we will use $b_{i,j}$ to denote $a^{(t-1)}_{i,j}$.  Then, we have

    \begin{align*}
        \frac{\sum_{\sigma \in \mathcal{C}} \prod_{i = t+1}^{t+k} a_{i, \sigma(i)}}{\per(A_{S\backslash \{t+k\}, S \backslash \{t+k\}})} &=  \frac{\sum_{\sigma \in \mathcal{C}} \prod_{i = t+1}^{t+k} \left(b_{i, \sigma(i)} + \frac{b_{i,t} b_{t,\sigma(i)}}{b_{t,t}}\right)}{\per\left(B_{S\backslash \{t+k\}, S \backslash \{t+k\}} + \frac{b_{S\backslash \{t+k\}, t} b_{t, S\backslash \{t+k\}}}{b_{t,t}}\right)} \\
        &= \frac{\sum_{\sigma \in \mathcal{C}} \prod_{i = t+1}^{t+k} \left(b_{i, \sigma(i)} b_{t,t} + b_{i,t} b_{t,\sigma(i)}\right)}{b_{t,t} \per\left(b_{t,t} \cdot B_{S\backslash \{t+k\}, S \backslash \{t+k\}} +b_{S\backslash \{t+k\}, t} b_{t, S\backslash \{t+k\}}\right)}.
    \end{align*}
    We define the numerator as a polynomial in $b_{t,t}$ as follows:
    \begin{equation*}
        \sum_{\sigma \in \mathcal{C}} \prod_{i = t+1}^{t+k} \left(b_{i, \sigma(i)} b_{t,t} + b_{i,t} b_{\sigma(i), t}\right) := \sum_{\ell = 0}^{k} \alpha_{\ell} \cdot b_{t,t}^{k-\ell}.
    \end{equation*}
    Without loss of generality, assume that $k$ is even. The proof for odd $k$ follows similarly. Let $D$ denote the denominator, we will show that 
    \begin{align*}
        \sum_{\ell \in \{0, 2, \ldots, k-2, k\}} \alpha_{\ell} \cdot b_{t,t}^{k-\ell} &\leq B(n, k, t-1) \ \cdot D \\
        \sum_{\ell \in \{1, 3, \ldots, k-3, k-1\} } \alpha_{\ell} \cdot b_{t,t}^{k-\ell} &\leq B(n, k+1, t-1) \cdot D.
    \end{align*}
    Combining the two equations gives
    \begin{align*}
        \sum_{\ell = 0}^{k} \alpha_{\ell} \cdot b_{t,t}^{k-\ell} \leq \left( B(n, k, t-1) + B(n, k+1, t-1)\right) \cdot D = B(n, k, t)  \cdot D,
    \end{align*}
where we used $B(n, a, b-1) + B(n, a+1, b-1) = B(n, a, b)$.

We will also expand the denominator $D$ also as a polynomial in $b_{t,t}$ as follows:
 \begin{align*}
        D = b_{t,t} \cdot \per\left(b_{t,t} \cdot B_{S\backslash \{t+k\}, S \backslash \{t+k\}} +b_{S\backslash \{t+k\}, t} b_{t, S\backslash \{t+k\}} \right) := \sum_{\ell = 0}^k \beta_\ell \cdot b_{t,t}^{k-\ell +1}.
    \end{align*}

    We will now define $\beta_{\ell}$ using the notion of cycle covers in a certain graph. 
    Consider a directed complete graph on the vertex set $S \cup \{t\}$. For $\ell \in \{0, \ldots, k\}$ and $X \subseteq S$, let $\mathcal{CC}_{\geq \ell}(X)$ denote the set of cycle covers of $X \cup \{t\}$ in this graph satisfying the following properties:
    \begin{itemize}
        \item every vertex in $X$  has exactly one incoming and one outgoing edge,
        \item there are exactly $\ell$ cycles containing $t$, and
        \item $t$ does not have a self-loop.
    \end{itemize}
    
    Then
    \begin{equation*}
        \beta_\ell = \ell! \cdot \sum_{C \in \mathcal{CC}_{\geq \ell}(S \backslash \{t+k\})} \prod_{(i\rightarrow j) \in C} b_{i,j}.
    \end{equation*}
    Here, $\ell!$ accounts for the fact that given cover $C \in \mathcal{CC}_{\geq \ell}(S)$, every matching of the $\ell$ incoming edges at $t$ to the $\ell$ outgoing edges at $t$ in $C$ gives a valid permutation on $S$.
    
    We will now relate $\alpha_\ell$ with $\beta_\ell$ using a different notion of cycle cover. For a fixed $\ell \in \{0, \ldots, k\}$ and $X \subseteq S$, let $\mathcal{CC}_{=\ell}(X)$ denote the set of all cycle covers of $X \cup \{t\} $ satisfying the following properties:
    \begin{itemize}
        \item there are no self-loops,
        \item every vertex in $X$ has exactly one incoming and one outgoing edge, and
        \item there are exactly $\ell$ cycles and each cycle contains $t$.
    \end{itemize}

  Then $ \alpha_{\ell}$ can be stated as
    \begin{align}
        \alpha_{\ell} = (\ell-1)!\sum_{C\in \mathcal{CC}_{=\ell}(S)} \prod_{(i\rightarrow j) \in C} b_{i,j} . \label{eq:alpha}
    \end{align}
    Here, we have an $(\ell-1)!$ term because the number of cyclic permutations in $\mathcal{C}$ which can lead to a cover $C \in \mathcal{CC}_{=\ell}$ is equal to the number of ways to arrange $\ell$ elements in a cycle, since $t$ has $\ell$ incoming and outgoing edges in $C$.

    Observe that there is a unique way to decompose any $C \in \mathcal{CC}_{\geq \ell}(S)$ or $C \in \mathcal{CC}_{=\ell}(S)$ into edge-disjoint cycles. So, for a cycle cover $C \in \mathcal{CC}_{=\ell}(S)$, let $X$ denote the vertex set of the cycle in $C$ containing $t+k$. We can re-write the equation \eqref{eq:alpha} in terms of $X$ as follows.
   \begin{align}
        \alpha_{\ell} &= (\ell-1)! \sum_{X \subseteq S: t+k \in X} \left(\sum_{c \text{ cycle on } X} \prod_{(i\rightarrow j) \in c} b_{i,j}\right) \cdot \left(   \sum_{C' \in \mathcal{CC}_{=\ell-1}(S \backslash X)} \prod_{(i\rightarrow j) \in C'} b_{i,j} \right). \label{eq:recursion}
    \end{align}

By the inductive hypothesis, we have 
\begin{equation*}
    \left(\sum_{c \text{ cycle on } X} \prod_{(i\rightarrow j) \in c} b_{i,j}\right) \leq B(n, |X|, t-1) \cdot \per(B_{X \backslash \{t+k\}, X \backslash \{t+k\}}).
\end{equation*}
Substituting this in equation \eqref{eq:recursion} gives
    \begin{align*}
        \alpha_{\ell} 
        &\leq (\ell-1)! \sum_{X \subseteq S: t+k \in X} B(n, |X|, t-1) \cdot \per(B_{X \backslash \{t+k\}, X \backslash \{t+k\}})  \cdot \left(   \sum_{C' \in \mathcal{CC}_{=\ell-1}(S \backslash X)} \prod_{(i\rightarrow j) \in C'} b_{i,j} \right).
    \end{align*}
For $\ell \neq 1$, $X$ can contain at most $k$ vertices and for $\ell = 1$, $|X| = k+1$, so for any even $\ell$, we have 
 \begin{align*}
        \alpha_{\ell} 
        &\leq (\ell-1)! \cdot B(n, k, t-1)  \sum_{X \subseteq S: t+k \in X}  \per(B_{X \backslash \{t+k\}, X \backslash \{t+k\}})  \cdot \left(   \sum_{C' \in \mathcal{CC}_{\ell-1}(S \backslash X)} \prod_{(i\rightarrow j) \in C'} b_{i,j} \right).
    \end{align*}

For a set $X$, 
\begin{align*}
    \per(B_{X \backslash \{t+k\}, X \backslash \{t+k\}}) &= \sum_{\sigma \in \mathcal{S}_{|X|-1}} \prod_{i \in X\backslash \{t+k\}} b_{i, \sigma(i)} \\
    &= b_{t,t} \sum_{C \in \mathcal{CC}_{\geq 0}(X \backslash \{t+k\})} \prod_{(i\rightarrow j) \in C} b_{i,j}+ \sum_{C \in \mathcal{CC}_{\geq 1}(X\backslash \{t+k\})} \prod_{(i\rightarrow j) \in C} b_{i,j}.
\end{align*}
The second equation follows by dividing the permutations on $X$ based on whether they contain the self-loop $t \rightarrow t$.

Therefore, 
\begin{align}
        \alpha_{\ell} 
        &\leq (\ell-1)! \cdot B(n, k, t-1)  \sum_{X \subseteq S: t+k \in X}     \left( b_{t,t} \sum_{C \in \mathcal{CC}_{\geq 0}(X \backslash \{t+k\})} \prod_{(i\rightarrow j) \in C} b_{i,j}\right)  \sum_{C' \in \mathcal{CC}_{=\ell-1}(S \backslash X)} \prod_{(i\rightarrow j) \in C'} b_{i,j} \notag\\
        &+(\ell-1)! \cdot B(n, k, t-1)  \sum_{X \subseteq S: t+k \in X}    \left( \sum_{C \in \mathcal{CC}_{\geq 1}(X\backslash \{t+k\})} \prod_{(i\rightarrow j) \in C} b_{i,j}\right)  \sum_{C' \in \mathcal{CC}_{=\ell-1}(S \backslash X)} \prod_{(i\rightarrow j) \in C'} b_{i,j} . \label{eq:big}
    \end{align}
    We will bound the two summations in equation \eqref{eq:big} separately.

    Note that for any $C\in \mathcal{CC}_{\geq 0}(X \backslash \{t+k\})$ and $C' \in \mathcal{CC}_{= \ell-1}(S\backslash X)$, given $C \cup C'$, there is only one way to reconstruct the set $X$: $X$ must contain $t$ and all vertices in cycles not containing $t$, i.e., vertices in $C$. Therefore, 
    \begin{align*}
         \sum_{X \subseteq S: t+k \in X}   &\left( b_{t,t} \sum_{C \in \mathcal{CC}_{\geq 0}(X \backslash \{ t+k\})} \prod_{(i\rightarrow j) \in C} b_{i,j}\right)  \sum_{C' \in \mathcal{CC}_{=\ell-1}(S \backslash X)} \prod_{(i\rightarrow j) \in C'} b_{i,j} \\
         &= b_{t,t} \sum_{C\in \mathcal{CC}_{\geq \ell-1}(S \backslash \{t+k\}) } \prod_{(i\rightarrow j) \in C} b_{i,j} = b_{t,t} \beta_{\ell-1}.
    \end{align*}
    However, given $C \cup C'$ with $C\in \mathcal{CC}_{\geq 1}(X \backslash \{t+k\})$ and $C' \in \mathcal{CC}_{= \ell-1}(S\backslash X)$, there are at most $\ell$ possibilities for set $X$: $X$ must contain all vertices in any cycle not passing through $t$, and $X$ must contain vertices from exactly one cycle passing through $t$. Since $C \cup C'$ contains $\ell$ cycles passing through $t$, we have
    \begin{align*}
         \sum_{X \subseteq S: t+k \in X}  \sum_{C \in \mathcal{CC}_{\geq 1}(X \backslash \{t+k\})} \prod_{(i\rightarrow j) \in C}& b_{i,j} \cdot  \sum_{C' \in \mathcal{CC}_{\ell-1}(S \backslash X)} \prod_{(i\rightarrow j) \in C'} b_{i,j} \\
         &\leq \ell\cdot  \sum_{C\in \mathcal{CC}_{\geq \ell}(S \backslash \{t+k\}) } \prod_{(i\rightarrow j) \in C} b_{i,j} 
         = \ell \cdot \beta_{\ell}.
    \end{align*}
    Plugging these bounds in equation \eqref{eq:big}, we get 
    \begin{align*}
        \alpha_{\ell} 
        &\leq B(n, k, t-1) \cdot  \left( b_{t,t} \beta_{\ell-1} + \beta_{\ell}\right).
    \end{align*}

    Summing over all even $\ell$'s, we have
    \begin{align}
        \sum_{\ell \in \{0, 2, \ldots, k-2, k\}} \alpha_{\ell} \cdot b_{t,t}^{k-\ell} &\leq B(n, k, t-1) \cdot \left(\beta_0 \cdot b_{t,t}^k + \sum_{\ell \in \{2, \ldots, k-2, k\}} \beta_{\ell-1} b_{t,t}^{k-\ell+1} +   \beta_{\ell} \cdot b_{t,t}^{k-\ell} \right) \notag\\
        &= B(n, k, t-1)  \cdot D. \label{eq:bounded-1}
    \end{align}

    Similarly, for odd $\ell$'s, we have
 \begin{align}
        \sum_{\ell \in \{1, 3, \ldots, k-1\}} \alpha_{\ell} \cdot b_{t,t}^{k-\ell} &\leq B(n, k+1, t-1) \cdot \left(\sum_{\ell \in \{1, 3, \ldots, k-1\}} \beta_{\ell-1} b_{t,t}^{k-\ell+1} +   \beta_{\ell} \cdot b_{t,t}^{k-\ell} \right) \notag\\
        &= B(n, k+1, t-1) \cdot D.\label{eq:bounded-2}
    \end{align}
    Summing equation \eqref{eq:bounded-1} and equation \eqref{eq:bounded-2} completes the proof.
\end{proof}

\subsection{Discussion and Open Questions}
We remark that the permanent process does not give an $O(1)^n$-approximation algorithm for general matrices. For example, the $n \times n$ all $1$'s matrix has permanent $n!$, but the product of diagonal entries returned by the permanent process is $2^{O(n^2)}$. Rather than a universal method, the permanent process should be viewed as a computationally inexpensive tool in the toolbox for bounding the permanent. While it does not provide optimal bounds for all matrices, for certain classes of structured matrices, it yields competitive upper bounds in regimes where other methods fail.

A natural direction for further research is a sharper characterization of settings in which the permanent process performs well. The characterization in Theorem \ref{thm:application} is far from tight. Similarly, can one identify the conditions under which the permanent process significantly improves on the product-of-row-sums bound, or when the two bounds are of the same order?

Another natural question is about the existence of a corresponding constructive lower bound procedure that complements the permanent process. Such a method would not only help certify the tightness of the upper bounds in specific regimes but could also lead to a more complete understanding of the asymptotic growth of permanents for structured matrix families.

\newpage
{\small
\bibliographystyle{alpha}
\bibliography{references}

\begin{thebibliography}{AGGS17}

\bibitem[AA11a]{AA11-boson-sampling}
Scott Aaronson and Alex Arkhipov.
\newblock The computational complexity of linear optics.
\newblock In {\em Proceedings of the Forty-Third Annual ACM Symposium on Theory of Computing}, STOC '11, page 333–342, New York, NY, USA, 2011. Association for Computing Machinery.

\bibitem[AA11b]{aaronson2011computational}
Scott Aaronson and Alex Arkhipov.
\newblock The computational complexity of linear optics.
\newblock In {\em Proceedings of the forty-third annual ACM symposium on Theory of computing}, pages 333--342, 2011.

\bibitem[AGGS17]{anari2017simply}
Nima Anari, Leonid Gurvits, Shayan~Oveis Gharan, and Amin Saberi.
\newblock Simply exponential approximation of the permanent of positive semidefinite matrices.
\newblock In {\em 2017 IEEE 58th Annual Symposium on Foundations of Computer Science (FOCS)}, pages 914--925. IEEE, 2017.

\bibitem[Bre73]{Bregman73}
L.~M. Bregman.
\newblock Some properties of nonnegative matrices and their permanents.
\newblock {\em Sov. Math., Dokl.}, 14:945--949, 1973.

\bibitem[ENG25]{ebrahimnejad2025approximability}
Farzam Ebrahimnejad, Ansh Nagda, and Shayan~Oveis Gharan.
\newblock On approximability of the permanent of psd matrices.
\newblock In {\em Proceedings of the 57th Annual ACM Symposium on Theory of Computing}, pages 625--630, 2025.

\bibitem[GS14]{GS-14}
Leonid Gurvits and Alex Samorodnitsky.
\newblock Bounds on the permanent and some applications.
\newblock In {\em 2014 IEEE 55th Annual Symposium on Foundations of Computer Science}, pages 90--99, 2014.

\bibitem[HKM98]{HKM-98}
Suk-Geun Hwang, Arnold~R. Kräuter, and T.S. Michael.
\newblock An upper bound for the permanent of a nonnegative matrix.
\newblock {\em Linear Algebra and its Applications}, 281(1):259--263, 1998.

\bibitem[HLLB08]{HLLB08-monomer-dimer}
Yan Huo, Heng Liang, Si-Qi Liu, and Fengshan Bai.
\newblock Computing monomer-dimer systems through matrix permanent.
\newblock {\em Phys. Rev. E}, 77:016706, Jan 2008.

\bibitem[LB04]{LB-04}
Heng Liang and Fengshan Bai.
\newblock An upper bound for the permanent of (0,1)-matrices.
\newblock {\em Linear Algebra and its Applications}, 377:291--295, 2004.

\bibitem[Mei23]{meiburg2023inapproximability}
Alexander Meiburg.
\newblock Inapproximability of positive semidefinite permanents and quantum state tomography.
\newblock {\em Algorithmica}, 85(12):3828--3854, 2023.

\bibitem[Min63]{Minc_binary63}
Henryk Minc.
\newblock Upper bounds for permanents of $\left( \{0,\,1\} \right)$-matrices.
\newblock {\em Bulletin of the American Mathematical Society}, 69:789--791, 1963.

\bibitem[Min84]{Minc_book-84}
Henryk Minc.
\newblock {\em Permanents}.
\newblock Encyclopedia of Mathematics and its Applications. Cambridge University Press, 1984.

\bibitem[MM65]{Marcus-Minc_perm65}
Marvin Marcus and Henryk Minc.
\newblock Permanents.
\newblock {\em The American Mathematical Monthly}, 72(6):577--591, 1965.

\bibitem[MN62]{marcus1962inequalities}
Marvin Marcus and Morris Newman.
\newblock Inequalities for the permanent function.
\newblock {\em Annals of Mathematics}, 75(1):47--62, 1962.

\bibitem[Sam08]{Samorodnitsky-08}
Alex Samorodnitsky.
\newblock An upper bound for permanents of nonnegative matrices.
\newblock {\em Journal of Combinatorial Theory, Series A}, 115(2):279--292, 2008.

\bibitem[Sch78]{Schrijver_Minc-78}
Alexander Schrijver.
\newblock A short proof of minc's conjecture.
\newblock {\em J. Comb. Theory A}, 25:80--83, 1978.

\bibitem[Sch98]{Schrijver_matching-98}
Alexander Schrijver.
\newblock Counting 1-factors in regular bipartite graphs.
\newblock {\em J. Comb. Theory Ser. B}, 72(1):122–135, January 1998.

\bibitem[Sch04]{scheel2004permanents}
Stefan Scheel.
\newblock Permanents in linear optical networks.
\newblock {\em arXiv preprint quant-ph/0406127}, 2004.

\bibitem[Sol00]{Soules_minc-00}
George~W. Solues.
\newblock Extending the minc-brègman upper bound for the permanent.
\newblock {\em Linear and Multilinear Algebra}, 47(1):77--91, 2000.

\bibitem[Sou03]{Soules_nonnegative-03}
George~W. Soules.
\newblock New permanental upper bounds for nonnegative matrices.
\newblock {\em Linear and Multilinear Algebra}, 51(4):319--337, 2003.

\bibitem[ST03]{shirai2003random}
Tomoyuki Shirai and Yoichiro Takahashi.
\newblock Random point fields associated with certain fredholm determinants i: fermion, poisson and boson point processes.
\newblock {\em Journal of Functional Analysis}, 205(2):414--463, 2003.

\bibitem[Val79]{valiant-permanet-hardness}
L.G. Valiant.
\newblock The complexity of computing the permanent.
\newblock {\em Theoretical Computer Science}, 8(2):189--201, 1979.

\bibitem[YP22]{yuan2022maximizing}
Chenyang Yuan and Pablo~A Parrilo.
\newblock Maximizing products of linear forms, and the permanent of positive semidefinite matrices.
\newblock {\em Mathematical Programming}, 193(1):499--510, 2022.

\end{thebibliography}
}

\appendix
\section{Omitted Proofs} \label{sec:omitted-proofs}

\begin{proof}[Proof of Theorem \ref{thm:Gaussian-determinant}]
	We prove using induction on $t$. The base case $t=1$ is trivial. For $t>1$, it is sufficient to prove the theorem for $j\geq t$ because for smaller $j$, the entry is determined at a smaller time step. The inductive step is to show that 
	\begin{align*}
		\frac{\det_A([t-1]+\{i\},[t-1]+\{j\})}{\det_A([t-1],[t-1])}&= \frac{\det_A([t-2]+\{i\}, [t-2]+\{j\})}{\det_A([t-2], [t-2])}\\&-\frac{\frac{\det_A([t-2]+\{i\}, [t-1])}{\det_A([t-2],[t-2])}\cdot \frac{\det_A([t-1], [t-2]+\{j\})}{\det_A([t-2], [t-2])}}{\frac{\det_A([t-1],[t-1])}{\det_A([t-2], [t-2])}}.
	\end{align*}
	Simplifying gives 
	\begin{align*}
		\operatorname{det}_A([t{-}1]{+}\{i\},[t{-}1]{+}\{j\}) \cdot \operatorname{det}_A([t{-}2],[t{-}2]) &= \operatorname{det}_A([t{-}2]{+}\{i\}, [t{-}2]{+}\{j\}) \cdot \operatorname{det}_A([t{-}1],[t{-}1]) \\
		& - \operatorname{det}_A([t{-}2]{+}\{i\}, [t{-}1]) \cdot \operatorname{det}_A([t{-}1], [t{-}2]{+}\{j\}).
	\end{align*}
	
	This is exactly the identity in Lemma \ref{lem:schur-determinant} with 
	\begin{align*}
		&B=A([t-2],[t-2]);\, y_1=A([t-2], \{t-1\}),\, y_2=A([t-2],\{j\});\\ &x_1^{\top}=A(\{t-1\},[t-2]),\, x_2^{\top}=A(\{i\}, [t-2]);\, w=A(\{t-1,i\},\{t-1,j\}). 
	\end{align*}
\end{proof}

\begin{proof}[Proof of Corollary \ref{cor:Gauss_elim-invar}]
	Using Theorem \ref{thm:Gaussian-determinant}, we have $a_{i,i}^{(n)}=\det_A([i],[i])/\det_A([i-1],[i-1])$. Substituting this gives $ \prod_{1\leq i\leq n}a_{i,i}^{(n)}=\det_A([n],[n])=\det(A)$.
\end{proof}

\begin{lemma}
	\label{lem:schur-determinant}
	For any matrix $B\in \RR^{d\times d}$, vectors $x_i,y_i \in \RR^{d\times 1}$, and scalars $w_{i,j}$ with $i,j \in \{1,2\}$, the following equality holds true:
	\begin{align}
		\label{eqn:determinant-uncrossing}
		\begin{vmatrix}
			B & y_1 & y_2 \\
			x_1^{\top}& w_{1,1} & w_{1,2}\\
			x_2^{\top}& w_{2,1} & w_{2,2}
		\end{vmatrix}\cdot |B|= \begin{vmatrix}
			B & y_1 \\
			x_1^{\top}& w_{1,1}
		\end{vmatrix}\cdot\begin{vmatrix}
			B & y_2 \\
			x_2^{\top}& w_{2,2}
		\end{vmatrix}- \begin{vmatrix}
			B & y_2 \\
			x_1^{\top}& w_{1,2}
		\end{vmatrix}\cdot\begin{vmatrix}
			B & y_1 \\
			x_2^{\top}& w_{2,1}
		\end{vmatrix}
	\end{align}
	
\end{lemma}
\begin{proof}
	Assume that $|B|\neq 0$. Using Schur's formula (see Theorem \ref{thm:schur_formula}), we have 
	\begin{align*}
		\begin{vmatrix}
			B & y_1 & y_2 \\
			x_1^{\top}& w_{1,1} & w_{1,2}\\
			x_2^{\top}& w_{2,1} & w_{2,2}
		\end{vmatrix} =|B|\cdot \det\begin{bmatrix}
			w_{1,1}-x_1^{\top}B^{-1}y_1 & w_{1,2}-x_{1}^{\top}B^{-1}y_2\\
			\\
			w_{2,1}-x_2^{\top}B^{-1}y_1 & w_{2,2}-x_{2}^{\top}B^{-1}y_2
		\end{bmatrix}   
	\end{align*}
	and $\begin{vmatrix}
		B & y_j \\
		x_i^{\top}& w_{i,j}
	\end{vmatrix}=|B|\cdot(w_{i,j}-x_i^{\top}B^{-1}y_{j})$ for $i,j \in \{1,2\}$. Substituting these and factoring out $|B|^2$ from both LHS and RHS of equation \eqref{eqn:determinant-uncrossing}, gives
	\begin{align*}
		\det\begin{bmatrix}
			w_{1,1}-x_1^{\top}B^{-1}y_1 & w_{1,2}-x_{1}^{\top}B^{-1}y_2\\
			\\
			w_{2,1}-x_2^{\top}B^{-1}y_1 & w_{2,2}-x_{2}^{\top}B^{-1}y_2
		\end{bmatrix} &= ( w_{1,1}-x_1^{\top}B^{-1}y_1)(w_{2,2}-x_{2}^{\top}B^{-1}y_2) \\&-(w_{1,2}-x_{1}^{\top}B^{-1}y_2)(w_{2,1}-x_2^{\top}B^{-1}y_1).
	\end{align*}
    The identity in equation \eqref{eqn:determinant-uncrossing} should hold true even when $|B|=0$ using continuity of the determinants with respect to the entries of the matrix. 
\end{proof}

\begin{proof}[Proof of Claim \ref{claim:exp-matrix}]
For ease of notation, we will use $c$ to denote $c_n$. We will establish the claim by induction on $t$. The claim is trivially true for $t=1$. 
    
    Let's assume $j \leq i$, the proof of $i < j$ follows similarly. The final value of the $i,j$-th entry after the permanent process is given by
    \begin{equation}
        a_{i,j}^{(j)} = a_{i,j} + \sum_{s=1}^{j-1}\frac{a_{i,s}^{(s)} a_{s,j}^{(s)}}{a_{s,s}^{(s)}}. \label{eq:claim-eq}
    \end{equation}
    By the inductive hypothesis, we have
    \begin{align*}
        a_{i,s}^{(s)} &= a_{i,s} \cdot\left(1 + \sum_{k=1}^{s-1} 2^{k-1}\cdot c^{-2\cdot k}  \right), \\
        a_{s,j}^{(s)} &= a_{s,j} \cdot\left(1 + \sum_{k=1}^{s-1} 2^{k-1}\cdot c^{-2\cdot k}  \right), \text{ and } \\
        a_{s,s}^{(s)} &= a_{s,s} \cdot\left(1 + \sum_{k=1}^{s-1} 2^{k-1}\cdot c^{-2\cdot k}  \right).
    \end{align*}
    Substituting these values in equation \eqref{eq:claim-eq} gives
    \begin{align*}
        a_{i,j}^{(t+1)} &= a_{i,j} + \sum_{s=1}^{j-1}\frac{a_{i,s} a_{s,j}}{a_{s,s}} \cdot  \left(1 + \sum_{k=1}^{s-1} 2^{k-1}\cdot c^{-2\cdot k}  \right).
    \end{align*}
    Using $a_{i,j} = c^{-|i-j|}$, the above equation becomes
    \begin{align*}
        a_{i,j}^{(t+1)} &= c^{-(i-j)} + \sum_{s=1}^{j-1} c^{-i-j+2s}\cdot  \left(1 + \sum_{k=1}^{s-1} 2^{k-1}\cdot c^{-2\cdot k}  \right)\\
        &= c^{-(i-j)} \left(1 + \sum_{s=1}^{j-1} c^{-2(j-s)}\cdot  \left(1 + \sum_{k=1}^{s-1} 2^{k-1}\cdot c^{-2\cdot k}  \right)\right) \\
        &= c^{-(i-j)} \left(1 + \sum_{s=1}^{j-1} c^{-2(j-s)}  +  c^{-2j}\sum_{s=1}^{j-1} \sum_{k=1}^{s-1} 2^{k-1}\cdot c^{-2\cdot k + 2s}  \right) \\
        &= c^{-(i-j)} \left(1 + \sum_{\ell=1}^{j-1} c^{-2\ell}  +  c^{-2j}\sum_{s=1}^{j-1} \sum_{\ell=1}^{s-1} 2^{s-\ell-1}\cdot c^{2\ell}  \right) \\
        &= c^{-(i-j)} \left(1 + \sum_{\ell=1}^{j-1} c^{-2\ell}  +  c^{-2j} \sum_{\ell=1}^{j-1} 2^{-\ell-1}c^{2\ell}  \sum_{s=\ell+1}^{j-1}  2^{s} \right) \\
        &= c^{-(i-j)} \left(1 + \sum_{\ell=1}^{j-1} c^{-2\ell}  + \sum_{\ell=1}^{j-1} c^{-2(j-\ell)} \cdot (2^{j-\ell-1}-1)  \right) \\
        &= c^{-(i-j)} \left(1 + \sum_{\ell=1}^{j-1} c^{-2\ell}  + \sum_{\ell=1}^{j-1} c^{-2\ell} \cdot (2^{\ell-1}-1)  \right) = c^{-(i-j)} \left(1 + \sum_{\ell=1}^{j-1} c^{-2\ell} \cdot 2^{\ell-1}  \right).
    \end{align*}
    
\end{proof}
\end{document}